\documentclass[aps,preprintnumbers,nofootinbib,twocolumn,superscriptaddress]{revtex4}

\usepackage{epsfig,graphics}
\usepackage{amsfonts,amssymb,amsmath}            
\usepackage{amsthm}
\usepackage{color}

\DeclareMathOperator{\md}{sh}

\newcommand{\Th}{^{\mathrm{th}}}
\newcommand{\id}{\openone}
\newcommand{\ot}{\otimes}
\newcommand{\comment}[1]{}
\newcommand{\eps}{\varepsilon}
\newcommand{\ket}[1]{|#1\rangle}
\newcommand{\bra}[1]{\langle #1|}
\newcommand{\ketbra}[2]{|#1\rangle\!\langle#2|}
\newcommand{\braket}[2]{\langle#1|#2\rangle}
\newcommand{\proj}[1]{|#1\rangle\!\langle #1|}

\newcommand*{\cA}{\mathcal{A}}
\newcommand*{\cB}{\mathcal{B}}

\newcommand*{\cH}{\mathcal{H}}

\newcommand*{\cL}{\mathcal{L}}

\newcommand*{\cS}{\mathcal{S}}

\newcommand*{\cX}{\mathcal{X}}
\newcommand*{\cY}{\mathcal{Y}}

\newcommand{\bef}{\rightsquigarrow} 
\newcommand*{\nbef}{\not\rightsquigarrow}

\newcommand*{\nuparrow}{\hspace{-0.18ex}{\scriptscriptstyle{\not}}\hspace{0.18ex}\uparrow}

\theoremstyle{plain}
\newtheorem{theorem}{Theorem}
\newtheorem{lemma}{Lemma}

\newtheorem{corollary}{Corollary}

\theoremstyle{definition}
\newtheorem{definition}{Definition}

\newcommand*{\wf}[1]{#1}
\newcommand*{\psiu}{\phi}

\begin{document}

\title{A system's wave function is uniquely determined by its
  underlying physical state}

\date{$28\Th$ December 2017}

\author{Roger \surname{Colbeck}}
\email[]{roger.colbeck@york.ac.uk}
\affiliation{Department of Mathematics, University of York, YO10 5DD, UK}
\author{Renato \surname{Renner}}
\email[]{renner@phys.ethz.ch}
\affiliation{Institute for Theoretical Physics, ETH Zurich, 8093
Zurich, Switzerland}

\begin{abstract}
  We address the question of whether the quantum-mechanical wave
  function $\Psi$ of a system is uniquely determined by any complete
  description $\Lambda$ of the system's physical state.  We show that
  this is the case if the latter satisfies a notion of ``free
  choice''. This notion requires that certain experimental
  parameters|those that according to quantum theory can be chosen
  independently of other variables|retain this property in the
  presence of $\Lambda$.  An implication of this result is that, among
  all possible descriptions $\Lambda$ of a system's state compatible
  with free choice, the wave function $\Psi$ is as objective as
  $\Lambda$.
\end{abstract}

\maketitle

\section{Introduction}

The quantum-mechanical wave function, $\wf{\Psi}$, has a clear
operational meaning, specified by the Born rule~\cite{Born1926}. It
asserts that the outcome $X$ of a measurement, defined by a family of
projectors $\{\Pi_x\}$, follows a distribution $P_X$ given by
$P_{X}(x)=\bra{\Psi}\Pi_x\ket{\Psi}$, and hence links the wave
function $\wf{\Psi}$ to observations.  However, the link is
probabilistic: even if $\wf{\Psi}$ is known to arbitrary precision, we
cannot in general predict $X$ with certainty.

In classical physics, such indeterministic predictions are always a
sign of incomplete knowledge.\footnote{For example, when we assign a
  probability distribution $P$ to the outcomes of a die roll, $P$ is
  not an objective property but rather a representation of our
  incomplete knowledge. Indeed, if we had complete knowledge,
  including for instance the precise movement of the thrower's hand,
  the outcome would be deterministic.} This raises the question of
whether the wave function $\wf{\Psi}$ associated to a system
corresponds to an \emph{objective} property of the system, or whether
it should instead be interpreted \emph{subjectively}, i.e., as a
representation of our (incomplete) knowledge about certain underlying
objective attributes. Another alternative is to deny the existence of
the latter, i.e., to give up the idea of an underlying reality
completely.

Despite its long history, no consensus about the interpretation of the
wave function has been reached. A subjective interpretation was, for
instance, supported by the famous argument of Einstein, Podolsky and
Rosen~\cite{EPR1935} (see also~\cite{Einstein1935}) and, more
recently, by information-theoretic
considerations~\cite{Jaynes1990,Caves2002,Spekkens2007}. The opposite
(objective) point of view was taken, for instance, by Schr\"odinger (at
least initially), von Neumann, Dirac, and
Popper~\cite{vonNeumann1955,Dirac1958,Popper1967}.

To turn this debate into a more technical question, one may consider
the following gedankenexperiment: Assume you are provided with a set
of variables $\Lambda$ that are intended to describe the physical
state of a system. Suppose, furthermore, that the set $\Lambda$ is
\emph{complete}, i.e., there is nothing that can be added to $\Lambda$
to increase the accuracy of any predictions about the outcomes of
measurements on the system. If you were now asked to specify the wave
function $\wf{\Psi}$ of the system, would your answer be unique?

If so then $\wf{\Psi}$ is a function of the variables $\Lambda$ and
hence as objective as $\Lambda$.  The model defined by $\Lambda$ would
then be called \emph{$\Psi$-ontic}~\cite{HarSpe2010}. Conversely, the
existence of a complete set of variables $\Lambda$ that does not
determine the wave function $\wf{\Psi}$ would mean that $\wf{\Psi}$
cannot be interpreted as an objective property. $\Lambda$ would then
be called \emph{$\Psi$-epistemic} (see Fig.~\ref{fig_wavefunction}).\footnote{Note that the existence or
  non-existence of $\Psi$-epistemic theories is also relevant in the
  context of simulating quantum systems.  Here $\Lambda$ can be
  thought of as the internal state of a computer performing the
  simulation, and one would ideally like that storing $\Lambda$
  requires significantly fewer resources than would be required to
  store $\Psi$. However, a number of existing results already cast
  doubt on this possibility (see, for
  example,~\cite{Hardy,Montina08,Montina12}).}
  
\begin{figure}
\includegraphics[clip=true,trim=7.7cm 5.9cm 3.5cm 5.2cm,width=0.5\textwidth]{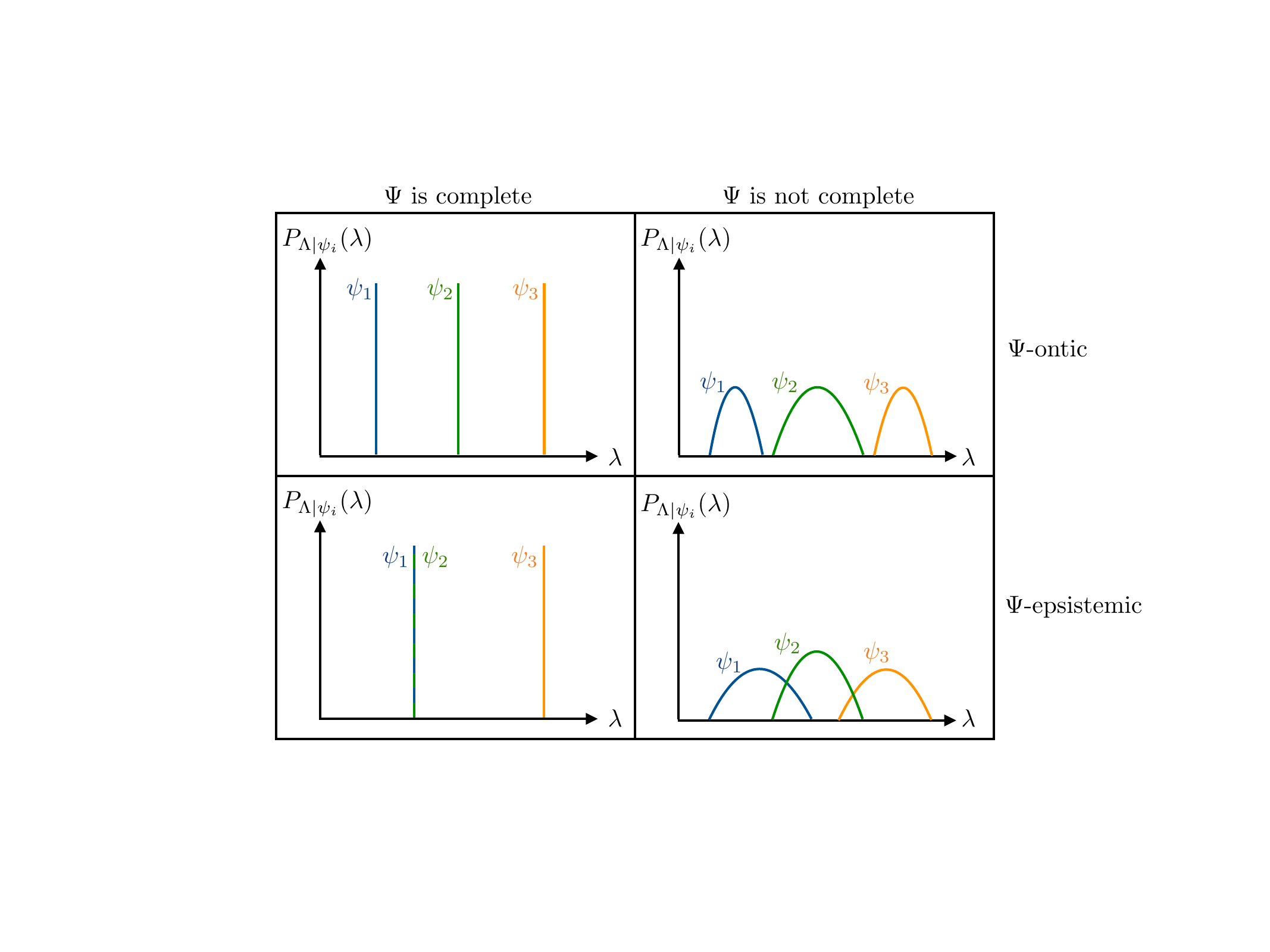}
\caption{\emph{The different possible roles of the wave function
    $\Psi$.} A model that uses a variable $\Lambda$ to describe a
  system's physical state can be either $\Psi$-ontic or
  $\Psi$-epistemic, depending on whether or not the wave function
  $\Psi$ is uniquely determined by $\Lambda$ (which takes values
  denoted by~$\lambda$).  Conversely, the relevant parts of $\Lambda$
  may be determined by $\Psi$, in which case $\Psi$ is complete. Using
  free choice (with respect to an appropriate causal order),
  \cite{ColRen2011} rules out the right column, \cite{ColRen2012}
  rules out the bottom left case, and the present paper (as well
  as~\cite{Pusey2012}, based on different assumptions) rules out the
  bottom row.}
\label{fig_wavefunction}
\end{figure}

In a seminal paper~\cite{Pusey2012}, Pusey, Barrett and Rudolph
showed that any complete model $\Lambda$ is $\wf{\Psi}$-ontic if it
satisfies an assumption, termed ``preparation independence''. It
demands that $\Lambda$ consists of separate variables for each
subsystem, e.g., $\Lambda=(\Lambda_A,\Lambda_B)$ for two subsystems
$S_A$ and $S_B$, and that these are statistically independent, i.e.,
$P_{\Lambda_A\Lambda_B}=P_{\Lambda_A}P_{\Lambda_B}$,
whenever the joint wave function $\wf{\Psi}$ of the total system has
product form, i.e., $\wf{\Psi}=\wf{\Psi_A}\ot\wf{\Psi_B}$.

Here we show that the same conclusion can be reached without imposing
any internal structure on $\Lambda$.  In more detail, our argument
relies on the concept of free choice, which can only be defined with
reference to an ordering, called here a \emph{causal
  order}\footnote{This should not be confused with a \emph{causal structure} as used in e.g.~\cite{Pearl}.}.  More precisely, we
prove that $\Psi$ is a function of any complete set of variables that
are compatible with free choice with respect to the causal order of
Figure~\ref{fig_causal} (see later for more details).  This is stated
as Corollary~\ref{cor_main}.  The free choice assumption used captures
the idea that experimental parameters, e.g., which state to prepare or
which measurement to carry out, can be chosen independently of all
other information (relevant to the experiment), except for information
that is created after the choice is made, e.g., measurement outcomes.
While this notion is implicit in quantum theory, we demand that it
also holds in the presence of $\Lambda$.\footnote{Free choice of
  certain variables is also implied by the preparation independence
  assumption used in~\cite{Pusey2012}, as discussed below.}

The proof of our result is inspired by our earlier
work~\cite{ColRen2012} in which we observed that the wave function
$\wf{\Psi}$ is uniquely determined by any complete set of
variables~$\Lambda$, provided that $\wf{\Psi}$ is itself complete (in
the sense described above). Together with the result
of~\cite{ColRen2011}, in which we showed that $\wf{\Psi}$ is complete,
we can conclude that the wave function $\wf{\Psi}$ is uniquely
determined by $\Lambda$.

The difference in the present work is that we can circumvent one of
the aspects of quantum theory required by the argument
in~\cite{ColRen2011}.  In particular, here we prove that $\wf{\Psi}$
is determined by $\Lambda$ without requiring that any quantum
measurement on a system corresponds to a unitary evolution of an
extended system.  Being based on weaker assumptions, the resulting
no-go theorem is stronger. Furthermore, the argument that the wave
function $\wf{\Psi}$ is complete is quite involved and a beneficial
feature of the present work is that we circumvent
it\footnote{Note, however, that the assumptions used in this work do
  not allow us to conclude that $\wf{\Psi}$ is complete.}.

\section{The Uniqueness Theorem}

Our argument refers to an experimental setup where a particle emitted
by a source decays into two, each of which is directed towards one of
two measurement devices (see Fig.~\ref{fig_setup}). The measurements
that are performed depend on parameters $A$ and $B$, and their
respective outcomes are denoted $X$ and $Y$.

\begin{figure}
\includegraphics[clip=true,trim=5cm 1.9cm 5cm 1cm,width=0.5\textwidth]{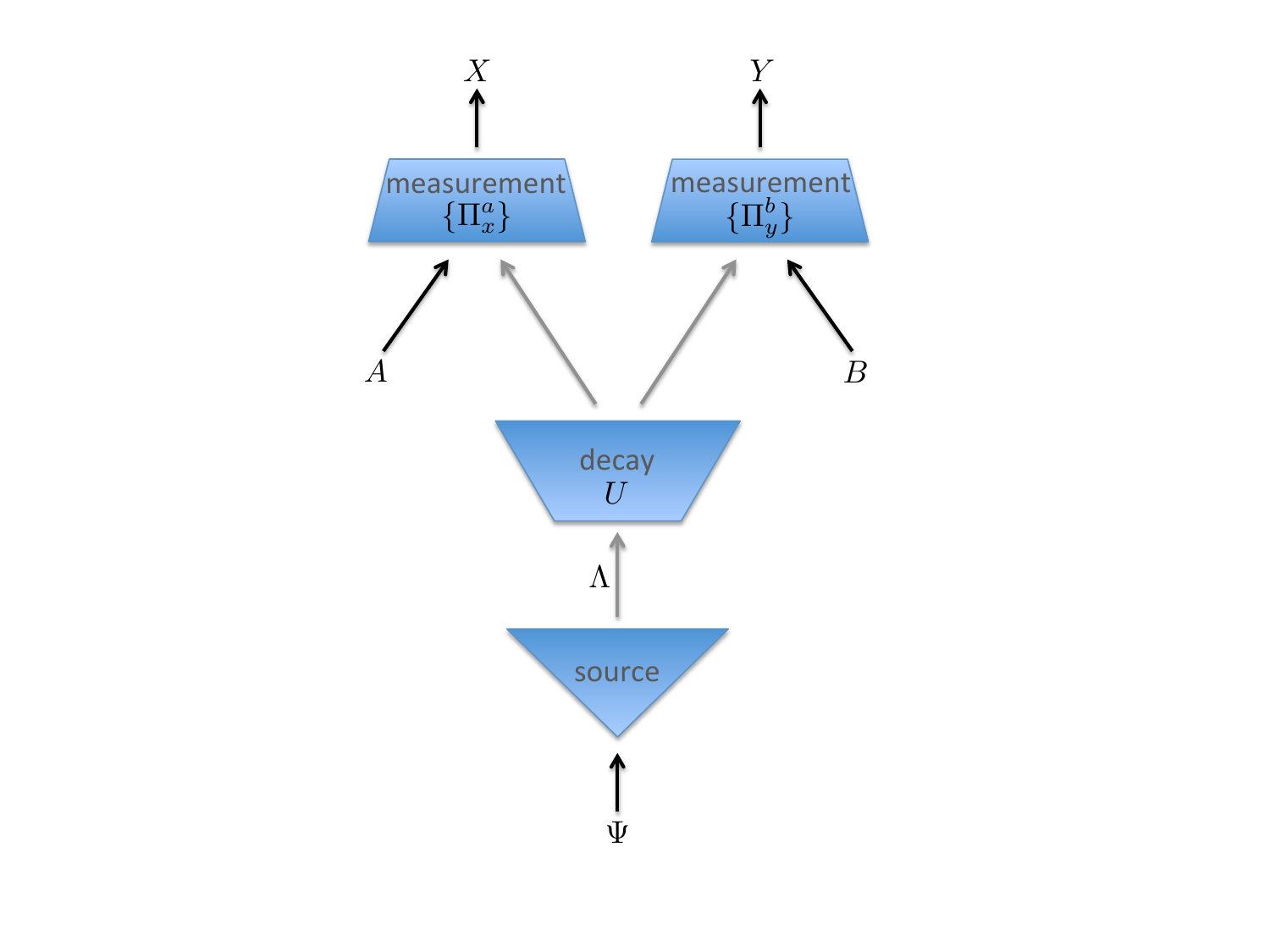}
\caption{\emph{The experimental setup.} The proof of the uniqueness theorem relies on a thought experiment where a source takes as input a description of a wave function $\Psi$ and prepares a particle in a corresponding state (which, in a general model, is described by a variable $\Lambda$). The particle then decays into two parts, which are measured at separate locations. $A$ and $B$ determine the measurements that are applied to the two parts, and $X$ and $Y$ are the respective outcomes.}
\label{fig_setup}
\end{figure}

Quantum theory allows us to make predictions about these outcomes
based on a description of the initial state of the system, the
evolution it undergoes and the measurement settings. For our purposes,
we assume that the quantum state of each particle emitted by the
source is pure, and hence specified by a wave function\footnote{We
  consider it uncontroversial that a mixed state can be thought of as
  a state of knowledge.}.  As we will consider different choices for
this wave function, we model it as a random variable $\Psi$ that takes
as values unit vectors $\wf{\psi}$ in a complex Hilbert space
$\cH$. Furthermore, we take the decay to act like an isometry, denoted
$U$, from $\cH$ to a product space $\cH_A\ot\cH_B$. Finally, for
any choices $a$ and $b$ of the parameters $A$ and $B$, the
measurements are given by families of projectors
$\{\Pi^a_x\}_{x\in\cX}$ and $\{\Pi^b_y\}_{y\in\cY}$ on $\cH_A$ and
$\cH_B$, respectively. The Born rule, applied to this setting, now
asserts that the joint probability distribution of $X$ and $Y$,
conditioned on the relevant parameters, is given by
\begin{align} \label{eq_Born}
  P_{XY|AB\Psi}(x,y|a,b,\psi)=\bra{\psi}U^\dagger(\Pi^a_x\ot\Pi^b_y)U\ket{\psi} \ .
\end{align}

To model the system's ``physical state'', we introduce an additional
random variable $\Lambda$.  We do not impose any structure on
$\Lambda$ (in particular, $\Lambda$ could be a list of values). We
will consider predictions $P_{XY|AB\Lambda}(x,y|a,b,\lambda)$
conditioned on any particular value $\lambda$ of $\Lambda$,
analogously to the predictions based on $\Psi$ according to the Born
rule~\eqref{eq_Born}.

To define the notions of \emph{free choice} and \emph{completeness},
as introduced informally in the introduction, we take as motivation
that any experiment takes place in spacetime and therefore
has a \emph{causal order}\footnote{In previous work we sometimes
called this a \emph{chronological structure}~\cite{CR_book2}.}. 
For example, the measurement setting $A$ is chosen before the
measurement outcome $X$ is obtained. This may be modelled
mathematically by a preorder relation\footnote{A \emph{preorder
    relation} is a binary relation that is reflexive and transitive.},
denoted $\bef$, on the relevant set of random variables.  While our
technical claim does not depend on how the causal order is interpreted
physically, it is intuitive to imagine it being compatible with
relativistic spacetime. In this case, $A \bef X$ would mean that the
spacetime point where $X$ is accessible lies in the future light cone
of the spacetime point where the choice $A$ is made.

\begin{figure}
\includegraphics[width=0.15\textwidth]{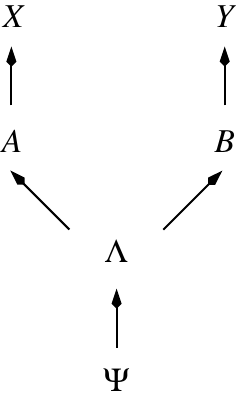}
\caption{\emph{The causal order.} Free choice is only
  well defined if one specifies a causal order, i.e., a preorder
  relation on the set of variables relevant to the experiment.
  The causal order we use is motivated by the arrangement of variables
  in the experiment depicted by Fig.~\ref{fig_setup} in relativistic
  space time.}
\label{fig_causal}
\end{figure}

For our argument we consider the causal order defined by
the transitive completion of the relations
\begin{align} \label{eq_causal} \Psi \bef \Lambda, \quad \Lambda \bef
  A, \quad \Lambda \bef B, \quad A \bef X, \quad B \bef Y
\end{align}
(cf.\ Fig.~\ref{fig_causal}). This reflects, for instance, that $\Psi$
is chosen at the very beginning of the experiment, and that $A$ and
$B$ are chosen later, right before the two measurements are carried
out. Note, furthermore, that $A \nbef Y$ and $B \nbef X$. With the
aforementioned interpretation of the relation in relativistic
spacetime, this would mean that the two measurements are carried out
at spacelike separation.

Using the notion of a causal order, we can now specify
mathematically what we mean by \emph{free choices} and by
\emph{completeness}.  We note that the two definitions below should be
understood as necessary (but not necessarily sufficient) conditions
characterising these concepts. Since they appear in the assumptions of
our main theorem, our result also applies to any more restrictive
definitions. We remark furthermore that the definitions are generic,
i.e., they can be applied to any set of variables equipped with a
preorder relation.\footnote{They are therefore different from notions
  used commonly in the context of Bell-type experiments, such as
  \emph{parameter independence} and \emph{outcome independence}. These
  refer explicitly to measurement choices and outcomes, whereas no
  such distinction is necessary for the definitions used here.}

\begin{definition} \label{def_free} When we say that a variable $A$ is
  a \emph{free choice from a set $\cA$ (w.r.t.\ a causal order)} this
  means that the support of $P_A$ contains $\cA$ and that
  $P_{A|A_{\nuparrow}} = P_A$ where $A_{\nuparrow}$ is the set of all
  random variables $Z$ (within the causal order) such that
  ${A \nbef Z}$.
\end{definition}

In other words, a choice $A$ is free if it is uncorrelated with any
other variables, except those that lie in the future of $A$ in the
causal order. For a further discussion and motivation of
this notion we refer to Bell's work~\cite{Bell2004} as well as
to~\cite{ColRen2013}.

Crucially, we note that Definition~\ref{def_free} is compatible with
the usual understanding of free choices within quantum theory. For
example, if we consider our experimental setup (cf.\
Fig.~\ref{fig_setup}) in ordinary quantum theory (i.e., where there is
no $\Lambda$), the initial state $\Psi$ as well as the measurement
settings $A$ and $B$ can be taken to be free choices w.r.t.~$\Psi \bef
A,\ \Psi \bef B,\ A \bef X,\ B \bef Y$ (which is the causal order
defined by Eq.~\ref{eq_causal} with $\Lambda$ removed).

\begin{definition}
  When we say that a variable $\Lambda$ is \emph{complete (w.r.t.\ a
    causal order)} this means that\footnote{In other words,
    $\Lambda_\downarrow\rightarrow\Lambda\rightarrow\Lambda_\uparrow$
    is a Markov chain. }
  \begin{align*}
    P_{\Lambda_\uparrow|\Lambda}=P_{\Lambda_\uparrow|\Lambda\Lambda_\downarrow} 
  \end{align*}
  where $\Lambda_\uparrow$ and $\Lambda_{\downarrow}$ denote the sets
  of random variables $Z$ (within the causal order) such
  that $\Lambda \bef Z$ and $Z \bef \Lambda$, respectively.
\end{definition}

Completeness of $\Lambda$ thus implies that predictions based on
$\Lambda$ about future values $\Lambda_\uparrow$ cannot be improved by
taking into account additional information $\Lambda_{\downarrow}$
available in the past.\footnote{Using statistics terminology, one may
  also say that $\Lambda$ is \emph{sufficient}  for $\Lambda_{\uparrow}$ given data $\Lambda_{\downarrow}$.}  Recall that this is
meant as a necessary criterion for completeness and that our
conclusions hold for any more restrictive definition. For example, one
may replace the set $\Lambda_{\uparrow}$ by the set of all values that
are not in the past of $\Lambda$.

We are now ready to formulate our main result as a theorem. Note that,
the assumptions of the theorem as well as its claim correspond to
properties of the joint probability distribution of $X$, $Y$, $A$,
$B$, $\Psi$ and $\Lambda$.

\begin{theorem} \label{thm_main} Let $\Lambda$ and $\Psi$ be random
  variables and assume that the support of $\Psi$ contains two wave
  functions, $\wf{\psi}$ and $\wf{\psi'}$, with
  $|\braket{\psi}{\psi'}|<1$. If for any isometry $U$ and
  measurements $\{\Pi^a_x\}_{x}$ and $\{\Pi^b_y\}_{y}$, parameterised
  by $a\in\cA$ and $b\in\cB$, there exist random variables $A$,
  $B$, $X$ and $Y$ such that
  \begin{enumerate}  
  \item \label{it_quantum} $P_{XY|AB\Psi}$ satisfies the Born rule
    \eqref{eq_Born}; 
  \item \label{it_free} $A$ and $B$ are free choices from $\cA$ and $\cB$,
    w.r.t.~\eqref{eq_causal};
  \item \label{it_completeness} $\Lambda$ is complete w.r.t.~\eqref{eq_causal} 
  \end{enumerate}
  then there exists a subset $\cL$ of the range of $\Lambda$
  such that $P_{\Lambda|\Psi}(\cL|\psi) = 1$ and
  $P_{\Lambda|\Psi}(\cL|\psi') = 0$.
\end{theorem}

The theorem asserts that, assuming validity of the Born rule and
freedom of choice, the values taken by any complete variable $\Lambda$
are different for different choices of the wave function $\Psi$. This
implies that $\Psi$ is indeed a function of $\Lambda$.

To formulate this implication as a technical statement, we consider an
arbitrary countable\footnote{The restriction to a countable set is due
  to our proof technique. We leave it as an open problem to determine
  whether this restriction is necessary.} set $\cS$ of wave functions
such that $|\braket{\psi}{\psi'}| < 1$ for any distinct elements
$\psi, \psi' \in \cS$.

\begin{corollary} \label{cor_main} Let $\Lambda$ and $\Psi$ be random
  variables with $\Psi$ taking values from the set $\cS$ of wave
  functions.  If the conditions of Theorem~\ref{thm_main} are
  satisfied then there exists a function $f$ such that $\Psi =
  f(\Lambda)$ holds almost surely.
\end{corollary}

The proof of this corollary is given in
Appendix~\ref{app_corollaryproof}.

\section{Proof of the Uniqueness Theorem}

The argument relies on specific wave functions, which depend on
parameters $d, k \in \mathbb{N}$ and $\xi \in [0, 1]$, with $k <
d$. They are defined as unit vectors on a product space ${\cH_A
  \ot \cH_B}$, where $\cH_A$ and $\cH_B$ are $(d+1)$-dimensional
Hilbert spaces equipped with an orthonormal basis
$\{\ket{j}\}_{j=0}^{d}$,\footnote{We use here the abbreviation
  $\ket{j} \ket{j}$ for $\ket{j} \ot \ket{j}$.}
\begin{align}
\wf{\psiu} &
=\frac{1}{\sqrt{d}}\sum_{j=0}^{d-1}\ket{j}\ket{j} \label{eq_psi} \\
    \wf{\psiu'} &=\frac{1}{\sqrt{k}} \Bigl( \xi \ket{0}\ket{0}  +
      \sum_{j=1}^{k-1}\ket{j}\ket{j} + \sqrt{1-\xi^2}
      \ket{d}\ket{d} \Bigr)\, . \label{eq_psip}
\end{align}

\begin{lemma}\label{lem_overlap}
  For any $0\leq\alpha<1$ there exist $k,d\in\mathbb{N}$ with $k<d$
  and $\xi\in[0,1]$ such that the vectors $\wf{\psiu}$ and
  $\wf{\psiu'}$ defined by~\eqref{eq_psi} and~\eqref{eq_psip} have
  overlap $\braket{\psiu}{\psiu'}=\alpha$.
\end{lemma}

\begin{proof}
  If $\alpha=0$, set $k=1$, $d=2$ and $\xi=0$. Otherwise, set $d\geq
  1/(1-\alpha^2)$, $k=\lceil\alpha^2d\rceil$ and $\xi=\alpha\sqrt{k d}-k+1$, so that $\xi\in [0,1]$ and 
  $\braket{\psiu}{\psiu'}=\alpha$.  Furthermore, the choice of $d$
  ensures that $\alpha^2 d+1\leq d$, which implies $k<d$.
\end{proof}

For any $n\in\mathbb{N}$, we consider projective measurements
$\{\Pi^a_x\}_{x \in \cX_d}$ and $\{\Pi^b_y\}_{y \in \cX_d}$ on $\cH_A$
and $\cH_B$, parameterised by
$a\in\cA_n\equiv\{0,2,4,\ldots,{2n-2}\}$ and
$b\in\cB_n\equiv\{1,3,5,\ldots,{2n-1}\}$, and with outcomes in
$\cX_d\equiv\{0,\ldots,{d}\}$. 

The outcomes $X$ and $Y$ will generally be correlated. To quantify
these correlations, we define\footnote{Note that the first sum
  corresponds to the probability that ${X \oplus 1 = Y}$, conditioned
  on $A=0$ and $B = 2n-1$. The terms in the second sum can be
  interpreted analogously. }
{\setlength{\jot}{-2pt}
\begin{multline*} 
I_{n,d}(P_{XY|AB})  \equiv 2n - \sum_{x=0}^{d-1} P_{XY|AB}(x, x\oplus 1|0,
2n-1) \\ -\!\!\!\! \sum_{\substack{a,b \\ |a-b|=1}}
\sum_{x=0}^{d-1} P_{XY|AB}(x, x|a, b) .
\end{multline*}
}

We set $\Pi^a_d=\Pi^b_d=\proj{d}$. For $x, y \in \{0, \ldots, d-1\}$,
the remaining projectors are chosen such that the value of the
quantity $I_{n,d}(P_{XY|AB})$ predicted by the Born rule when these
measurements are applied to the state $\psiu$ defined
by~\eqref{eq_psi} can be made arbitrarily small for large enough $n$.
More precisely they are chosen such that for
$P_{XY|AB}(x,y|a,b) = \bra{\psiu} {\Pi^a_x \ot \Pi^b_y}
\ket{\psiu}$ we have
\begin{align} \label{eq_Inbound}
  I_{n,d}(P_{XY|AB})\leq\frac{\pi^2}{6 n}   .
\end{align}
[For the details of how to choose the projectors and the derivation of
this bound, see Appendix~\ref{app_QM}.]

The next lemma shows that $I_{n,d}$ gives an upper bound on the
distance of the distribution $P_{X|A\Lambda}$ from a uniform
distribution over $\{0, \ldots, d-1\}$. The bound holds for any random
variable $\Lambda$, provided the joint distribution $P_{XY\Lambda|AB}$
satisfies certain conditions.

\begin{lemma}\label{lem_extended}
  Let $P_{XY AB \Lambda}$ be a distribution that satisfies
  $P_{X\Lambda|AB}=P_{X\Lambda|A}$, $P_{Y\Lambda|AB}=P_{Y\Lambda|B}$
  and $P_{A B \Lambda} = P_AP_BP_{\Lambda}$ with
  $\mathrm{supp}(P_A) \supseteq \cA_n$ and $\mathrm{supp}(P_B)
  \supseteq \cB_n$. Then
  \begin{align*}
   \int \mathrm{d}P_{\Lambda}(\lambda) \sum_{x=0}^{d-1} \bigl|P_{X|A\Lambda}(x|0,\lambda)-\frac{1}{d} \bigr| 
    \leq\frac{d}{2}I_{n,d}(P_{XY|AB}) \, .
  \end{align*}
\end{lemma}

(Although our proof deals with the general case, the main ideas can be
seen by working through the analogous argument in the slightly simpler
(but less general) case in which $\Lambda$ is discrete, so that
``$\int \mathrm{d}P_{\Lambda}(\lambda)$'' is replaced by
``$\sum_\lambda P_{\Lambda}(\lambda)$''.)

The proof of Lemma~\ref{lem_extended} is given in
Appendix~\ref{app_In}. It generalises an argument described
in~\cite{ColRen2011}, which is in turn based on work related to
chained Bell inequalities~\cite{Pearle1970,BraunsteinCaves1990} (see
also~\cite{BHK,BKP}).

We have now everything ready to prove the uniqueness theorem.

\begin{proof}[Proof of Theorem~\ref{thm_main}] Let
  $\alpha, \gamma\in\mathbb{R}$ such that
  $e^{i\gamma}\alpha=\braket{\psi}{\psi'}$. Furthermore, let
  $k, d, \xi$ be as defined by Lemma~\ref{lem_overlap}, so that
  $\braket{\psiu}{\psiu'}=\alpha$. Then there exists an isometry $U$
  such that $U\psi=\psiu$ and $U\psi'=e^{i\gamma}\psiu'$ (see
  Lemma~\ref{lem_isometry} of
  Appendix~\ref{app_additional}).\footnote{If $\cH$ has a larger
    dimension than $\cH_A\ot\cH_B$ (e.g., because $\cH$ is
    infinite dimensional) then we can consider an (infinite
    dimensional) extension of $\cH_B$, keeping the same notation for
    convenience.}  Now let $n \in \mathbb{N}$ and let $A$, $B$, $X$
  and $Y$ be random variables that satisfy the three conditions of the
  theorem for the isometry $U$ and for the projective measurements
  that satisfy the bound of Equation~\eqref{eq_Inbound}. According to
  the Born rule (Condition~\ref{it_quantum}), the distribution
  $P_{X Y|AB\psi} \equiv P_{XY|AB\Psi}(\cdot, \cdot|\cdot, \cdot,
  \psi)$
  conditioned on the choice of initial state $\Psi = \psi$ corresponds
  to the one considered in~\eqref{eq_Inbound}, i.e.,
\begin{align}
  I_{n,d}(P_{XY|AB\psi})\leq\frac{\pi^2}{6n}\ . \label{eq_Inzero} 
\end{align}

Note that
$P_{A|B\Psi}P_{Y\Lambda|AB\Psi}=P_{AY\Lambda|B\Psi}=P_{A|BY\Lambda
  \Psi} P_{Y\Lambda|B\Psi}$.  Freedom of choice
(Condition~\ref{it_free}) implies that
$P_{A|B\Psi}=P_{A|BY\Lambda\Psi}$.  It follows that
$P_{Y\Lambda|AB\Psi}=P_{Y\Lambda|B\Psi}$.  By a similar reasoning, we
also have $P_{X\Lambda|AB\Psi}=P_{X\Lambda|A\Psi}$. The freedom of
choice condition also ensures that
$P_{AB\Lambda|\Psi}=P_AP_BP_{\Lambda|\Psi}$ with
$\mathrm{supp}(P_A)\supseteq\cA_n$ and
$\mathrm{supp}(P_B)\supseteq\cB_n$. We can thus apply
Lemma~\ref{lem_extended} to give, with~\eqref{eq_Inzero},
\begin{align*}
\int\mathrm{d}P_{\Lambda|\psi}(\lambda)\sum_{x=0}^{d-1} \bigl| P_{X |A \Lambda \Psi}(x| 0,\lambda, \psi) -
  \frac{1}{d} \bigr|\leq\frac{d \pi^2}{12 n } \ .
\end{align*}
Considering only the term $x=k$ (recall that $k< d$) and noting that
the left hand side does not depend on $n$, we have
\begin{align*}
  \int\mathrm{d}P_{\Lambda|\psi}(\lambda)\bigl| P_{X |A \Lambda
    \Psi}(k| 0,\lambda, \psi) - \frac{1}{d} \bigr|=0
\end{align*}
(otherwise, by taking $n$ sufficiently large, we will get a
contradiction with the above).  Let $\cL$ be the set of all elements
$\lambda$ from the range of $\Lambda$ for which $P_{X |A \Lambda
  \Psi}(k| 0,\lambda, \psi)$ is defined and equal
to~$\frac{1}{d}$. The above implies that $P_{\Lambda|\Psi}(\cL|\psi) =
1$. Furthermore, completeness of $\Lambda$
(Condition~\ref{it_completeness}) implies that for any $\lambda \in
\cL$ for which $P_{X |A \Lambda \Psi}(k| 0,\lambda, \psi')$ is defined
\begin{align*}
  P_{X|A\Lambda\Psi}(k|0,\lambda,\psi')=P_{X|A\Lambda\Psi}(k|0,\lambda,\psi)=\frac{1}{d}\,
  .
\end{align*}
Thus, using $P_{\Lambda|A \Psi}=P_{\Lambda|\Psi}$ (which is implied
by the freedom of choice assumption, Condition~\ref{it_free}) and
writing $\delta_{\cL}$ for the indicator function, we have
\begin{align} \label{eq_PX}
  P_{X|A\Psi}(k|0,\psi')
  &=\int
  \mathrm{d}P_{\Lambda|\Psi}(\lambda|\psi')P_{X|A\Lambda\Psi}(k|0,\lambda,\psi')
 \\
  &\geq\int \delta_{\cL}(\lambda)
  \mathrm{d}P_{\Lambda|\Psi}(\lambda|\psi')P_{X|A\Lambda\Psi}(k|0,\lambda,\psi')
  \nonumber \\
  &=\frac{1}{d}\int \delta_{\cL}(\lambda) \mathrm{d}P_{\Lambda|\Psi}(\lambda|\psi')
  = \frac{1}{d} P_{\Lambda|\Psi}(\cL|\psi')  \nonumber
\, .
\end{align}

However, because the vector $e^{i \gamma}\psiu'=U\psi'$ has no
overlap with $\ket{k}$ (because $k<d$) and because the measurement
$\{\Pi^a_x\}_{x\in\cX_d}$ for $a=0$ corresponds to projectors along
the $\{\ket{x}\}_{x=0}^{d}$ basis, we have $P_{X|A\Psi}(k|0,\psi') =
0$ by the Born rule (Condition~\ref{it_quantum}). Inserting this
in~\eqref{eq_PX} we conclude that $P_{\Lambda|\Psi}(\cL|\psi') = 0$.
\end{proof}

\section{Discussion}

It is interesting to compare Theorem~\ref{thm_main} to the result
of~\cite{Pusey2012}, which we briefly described in the
introduction. The latter is based on a different experimental setup,
where $n$ particles with wave functions $\Psi_1,\ldots,\Psi_n$, each
chosen from a set $\{\psi,\psi'\}$, are prepared independently at $n$
remote locations. The $n$ particles are then directed to a device
where they undergo a joint measurement with outcome $Z$.

The main result of~\cite{Pusey2012} is that, for any variable
$\Lambda$ that satisfies certain assumptions, the wave functions
$\Psi_1, \ldots, \Psi_n$ are determined by $\Lambda$.  One of these
assumptions is that $\Lambda$ consists of $n$ parts, $\Lambda_1,
\ldots, \Lambda_n$, one for each particle. To state the other
assumptions and compare them to ours, it is useful to consider the
causal order defined by the transitive completion of the
relations\footnote{Note that this causal order captures the
  aforementioned experimental setup. In particular, we have $\Psi_i
  \nbef \Lambda_j$ for $i \neq j$, reflecting the idea that the $n$
  particles are prepared in separate isolated devices.}
\begin{align} \label{eq_causal2} \Psi_i \bef \Lambda_i \, \text{
    ($\forall \, i$)}, \quad (\Lambda_1, \ldots, \Lambda_n) \bef \Lambda,
  \quad \Lambda \bef Z \ .
\end{align}
It is then easily verified that the assumptions of~\cite{Pusey2012} imply
the following:
\begin{enumerate}
 \item $P_{Z|\Psi_1 \cdots \Psi_n}$ satisfies the Born rule;
 \item $\Psi_1, \ldots, \Psi_n$ are free choices from $\{\psi,
   \psi'\}$ w.r.t.~\eqref{eq_causal2};
 \item $\Lambda$ is complete w.r.t.~\eqref{eq_causal2}.
\end{enumerate}
These conditions are essentially in one-to-one correspondence with the
assumptions of Theorem~\ref{thm_main}.\footnote{The choice of a
  measurement setting may be encoded into the state of an extra system
  that is fed into a fixed measurement device. We hence argue that
  there is no conceptual difference between the free choice of a
  state, as implied by the assumptions of~\cite{Pusey2012} (in
  particular, preparation independence), and the free choice of a
  measurement setting, as assumed in Theorem~\ref{thm_main}.}  The
main difference thus concerns the modelling of the physical state
$\Lambda$, which in the approach of~\cite{Pusey2012} is assumed to
have an internal structure.  A main goal of the present work was to
avoid using this assumption (see also~\cite{Hardy_PBR,Aaronson_PBR}
for alternative arguments).

We conclude by noting that the assumptions of Theorem~\ref{thm_main}
and Corollary~\ref{cor_main} may be weakened. For example, the
independence condition that is implied by free choice may be replaced
by a partial independence condition along the lines considered
in~\cite{CR_free}.  An analogous weakening was given
in~\cite{Hall_PBR,SF} regarding the argument of~\cite{Pusey2012}. More
generally, recall that all our assumptions are properties of the
probability distribution $P_{XYAB\Psi\Lambda}$. One may therefore
replace them by relaxed properties that need only be satisfied for
distributions that are $\eps$-close (in total variation distance) to
$P_{XYAB\Psi\Lambda}$. (For example, the Born rule may only hold
approximately.) It is relatively straightforward to verify that the
proof still goes through, leading to the claim that $\Psi =
f(\Lambda)$ holds with probability at least $1-\delta$, with $\delta
\to 0$ in the limit where $\eps \to 0$.

Nevertheless, none of the three assumptions of Theorem~\ref{thm_main}
can be dropped without replacement.  Indeed, without the Born rule,
the wave function $\Psi$ has no meaning and could be taken to be
independent of the measurement outcomes $X$. Furthermore, a recent
impossibility result~\cite{Lewis2012} implies that the analogous
theorem with the second assumption omitted does not hold. It also
implies that the statement of Theorem~\ref{thm_main} cannot hold for a
setting with only one single measurement. This means that there exist
$\Psi$-epistemic theories compatible with the remaining
assumptions. However, in this case, it is still possible to exclude a
certain subclass of such theories, called \emph{maximally
  $\Psi$-epistemic} theories~\cite{Maroney} (see
also~\cite{LeiferMaroney}). Finally, completeness of $\Lambda$ is
necessary because, without it, $\Lambda$ could be set to a constant,
in which case it clearly cannot determine $\Psi$.

\acknowledgments

We thank Omar Fawzi, Michael Hush, Matt Leifer, Matthew Pusey and Rob
Spekkens for useful discussions. We are also grateful to Giorgos
Eftaxias for discussions that led to the discovery of an important
omission in an earlier version. Research leading to these results was
supported by the Swiss National Science Foundation (through the
National Centre of Competence in Research \emph{Quantum Science and
  Technology} and grant No.~200020-135048), the CHIST-ERA project
DIQIP, and the European Research Council (grant No.~258932).

\bibliographystyle{naturemag}

\appendix

\section{Proof of Corollary~\ref{cor_main}} \label{app_corollaryproof}

For any distinct $\psi, \psi' \in \cS$, let $\cL_{\psi, \psi'}$ be the
set defined by Theorem~\ref{thm_main}, i.e.,
\begin{align*}
  P_{\Lambda|\Psi}(\cL_{\psi, \psi'}|\psi) & = 1 \\
  P_{\Lambda|\Psi}(\cL_{\psi, \psi'}|\psi') & = 0 \ ,
\end{align*}
and for any $\psi \in \cS$ define the (countable) intersection
$\cL_\psi \equiv \bigcap_{\psi' \in \cS \setminus \{\psi\}} \cL_{\psi,
  \psi'}$.  This satisfies
\begin{align*}
  P_{\Lambda | \Psi}(\cL_{\psi}|\psi') = \begin{cases} 1 & \text{if
      $\psi = \psi'$} \\ 0 & \text{otherwise.} \end{cases}
\end{align*}
(Here we have used that for any probability distribution $P$ and for
any events $L, L'$, $P(L) = P(L') = 1$ implies that $P(L \cap L') =
1$.)  

  To define the function $f$, we specify the inverse sets
  \begin{align*}
    f^{-1}(\psi) = \cL_\psi \setminus \bigl(\bigcup_{\psi' \in
      \cS \setminus \{\psi\}} \cL_{\psi'} \bigr) \ .
  \end{align*}
  The function $f$ is well defined on
  $\bigcup_{\psi\in\cS}f^{-1}(\psi)$ because, by construction, the
  sets $f^{-1}(\psi)$ are disjoint for different $\psi \in
  \cS$. Furthermore, it follows from the above that for any $\psi \in
  \cS$
  \begin{align*}
    P_{\Lambda|\Psi}(f^{-1}(\psi)| \psi) = 1 \ .
  \end{align*}
  This implies that $f(\Lambda) = \Psi$ holds with probability $1$
  conditioned on $\Psi = \psi$. The assertion of the corollary then
  follows because this is true for any $\psi\in\cS$.  \qed

\section{Quantum correlations} \label{app_QM}
The aim of this appendix is to derive the bound~\eqref{eq_Inbound}
used in the proof of the uniqueness theorem.

Note that the state $\psiu$, defined by~\eqref{eq_psi}, has support on
$\bar{\cH}\ot\bar{\cH}$, where
$\bar{\cH}=\mathrm{span}\{\ket{0},\ket{1},\ldots,\ket{d-1}\}$. We will
choose the projectors $\Pi^a_x$ and $\Pi^b_y$ for $a\in\cA_n$ and
$b\in\cB_n$ and for $x,y\in\{0, \ldots, d-1\}$ to act on $\bar{\cH}$,
so can restrict to this subspace.

To define the projectors, we introduce the generalised Pauli operators
$\hat{X}_d=\sum_{l=0}^{d-1}\ketbra{l}{l\oplus 1}$, where $\oplus$
denotes addition modulo $d$, and
$\hat{Z}_d:=\sum_{j=0}^{d-1}e^{2\pi ij/d}\proj{j}$ as well as the
unitary $U_d:=\frac{1}{\sqrt{d}}\sum_{jk}e^{2\pi ijk/d}\ketbra{j}{k}$.
These operators satisfy $\hat{X}_d=U_d\hat{Z}_dU_d^{\dagger}$.  Our
construction is based on taking the $2n\Th$ root of these operators.
However, because there are many choices of $2n\Th$ root of a complex
number, we need to specify which $2n\Th$ root we mean.  We use a
slightly different choice for the measurements in $\cA_n$ and $\cB_n$.
These choices can be conveniently written by defining $\md_A[v]$ to
mean the number in $(-1/2,1/2]$ that is equal to $v+m$ for some
integer $m$ ($\md$ stands for ``shift'') and $\md_B[v]$ to mean the
number in $[-1/2,1/2)$ that is equal to $v+m$ for some integer
$m$. For $x\in\{0,\ldots,d-1\}$ and $a\in\{0,2,\ldots,{2n-2}\}$, the
projectors $\Pi^a_x$ are along the vectors
$\ket{\zeta_x^a}=U_dZ_{n,d}[a]U_d^\dagger\ket{x}$, where
\begin{align*}
 Z_{n,d}[a]:=\sum_{j=0}^{d-1}\exp\left[\pi i\md_A[j/d]\frac{a}{n}\right]\proj{j},
\end{align*}
while for $y\in\{0,\ldots,d-1\}$ and $b\in\{1,3,\ldots,{2n-1}\}$, the projectors
$\Pi^b_y$ are along the vectors $\ket{\zeta_y^b}=U_dZ'_{n,d}[b]U_d^\dagger\ket{y}$, where
\begin{align*}
  Z'_{n,d}[b]:=\sum_{j=0}^{d-1}\exp\left[\pi i\md_B[j/d]\frac{b}{n}\right]\proj{j}.
\end{align*}
Note that for each $k\in\{0,1,\ldots,2n-1\}$ and
$j,j'\in\{0,1,\ldots,d-1\}$ we have
$\braket{\zeta_j^k}{\zeta_{j'}^k}=\delta_{j,j'}$, so $\{\Pi_j^k\}_j$
is a projective measurement on $\bar{\cH}$.

The probability distribution that gives rise to the bound
in~\eqref{eq_Inbound} is obtained from a measurement of $\psiu$ with
respect to these projectors, i.e.,
$P_{XY|AB}(x,y|a,b)=|(\bra{\zeta_x^a}\bra{\zeta_y^b})\ket{\psiu}|^2$. We
are now going to show that
\begin{align} \label{eq_probneighbour}
\sum_x P_{XY|AB}(x,x|a,b) =\frac{\sin^2\frac{\pi}{2n}}{d^2\sin^2\frac{\pi}{2dn}}\, ,
\end{align}
for $|a-b|=1$, and
\begin{align} \label{eq_probborder}
\sum_x P_{XY|AB}(x,x\oplus
1|0,2n-1)=\frac{\sin^2\frac{\pi}{2n}}{d^2\sin^2\frac{\pi}{2dn}}\, .
\end{align}

For this it is useful to use the relation that for any operator $C$,
$(\id\ot C)\ket{\psiu}=(C^{\text{T}}\ot\id)\ket{\psiu}$, where
$C^{\text{T}}$ denotes the transpose of $C$ in the $\ket{i}$
basis. Thus, noting that $U_d^{\text{T}}=U_d$, we have
\begin{align*}
  (\bra{\zeta_x^a}\bra{\zeta_x^b})\ket{\phi}&=(\bra{\zeta_x^a}\bra{x})(\id\ot
                                              U_d(Z'_{n,d}[b])^\dagger
                                              U_d^\dagger)\ket{\phi}\\
&=\bra{\zeta_x^a}\bra{x}(U_d^\dagger(Z'_{n,d}[b])^\dagger
    U_d\ot\id)\ket{\phi}\\
&=\frac{1}{\sqrt{d}}\bra{\zeta_x^a}U_d^\dagger(Z'_{n,d}[b])^\dagger
    U_d\ket{x}\\
&=\frac{1}{\sqrt{d}}\bra{x}U_d(Z_{n,d}[a])^\dagger(U_d^{\dagger})^2(Z'_{n,d}[b])^\dagger
  U_d\ket{x}.
\end{align*}
Then, using
\begin{align*}
(U_d^{\dagger})^2&=\frac{1}{d}\sum_{jkm}e^{-2\pi
  ij(k+m)/d}\ketbra{k}{m}
=\sum_{k=0}^{d-1}\ketbra{k}{-k \oplus d}\, ,
\end{align*}
so that
\begin{align*}
  \bra{k}(U_d^\dagger)^2\ket{m}&=\delta_{k,0}\delta_{m,0}+\sum_{j=1}^{d-1}\delta_{k,j}\delta_{d-j,m}
\end{align*}
we obtain
\begin{align*}
(&\bra{\zeta_x^a}\bra{\zeta_x^b})\ket{\phi}=\frac{1}{d^{3/2}}\left(1+\phantom{\sum_{j=1}^{d-1}}\right.\\
&\left.\sum_{j=1}^{d-1}\exp\!\left[\pi
  i\left(2x\!-\!\md_A\!\left[\frac{j}{d}\right]\!\frac{a}{n}\!-\!\md_B\!\left[\frac{d-j}{d}\right]\!\frac{b}{n}\right)\right]\right)\!.
\end{align*}

We then note that 
$$\md_A[j/d]=\left\{\begin{array}{rcc}j/d&\ &0\leq j\leq d/2\\j/d-1&\ &d/2<j\leq d-1\end{array}\right.$$
and
$$\md_B[1-j/d]=\left\{\begin{array}{rcc}-j/d&\ &0\leq j\leq d/2\\1-j/d&\ &d/2<j\leq d-1\end{array}\right.$$
so we can divide the sum into two parts (absorbing the ``$1+$'' into the
first sum by allowing $j=0$):
\begin{align*}
(&\bra{\zeta_x^a}\bra{\zeta_x^b})\ket{\phi}=\frac{1}{d^{3/2}}\left(\sum_{j=0}^{\lfloor
    d/2\rfloor}\exp\left[-\pi i\frac{(a-b)j}{nd}\right]+\right.\\
&\left.\sum_{j=\lfloor d/2\rfloor+1}^{d-1}\exp\left[-\pi i\frac{(a-b)}{n}\left(\frac{j}{d}-1\right)\right]\right)\\
&=\frac{-1}{d^{3/2}}\exp\!\left[-\pi
  i\frac{(a\!-\!b)}{nd}(\lfloor\frac{d}{2}\rfloor+1)\right]\!\frac{1-\exp\!\left[\pi
  i\frac{(a-b)}{n}\right]}{1-\exp\!\left[-\pi
    i\frac{(a-b)}{nd}\right]} .
\end{align*}


We can then use $|1-e^{iy}|^2=4\sin^2\frac{y}{2}$ to obtain
\begin{align*}
|(\bra{\zeta_x^a}\bra{\zeta_x^b})\ket{\phi}|^2&=\frac{1}{d^3}\frac{\sin^2\frac{\pi(a-b)}{2n}}{\sin^2\frac{\pi(a-b)}{2dn}}\, .
\end{align*}
This is independent of $x$, so if we sum over all $x$ we obtain
\begin{align*}
\sum_xP_{XY|AB}(x,x|a,b)=\frac{\sin^2\frac{\pi(a-b)}{2n}}{d^2\sin^2\frac{\pi(a-b)}{2dn}}\, ,
\end{align*}
from which~\eqref{eq_probneighbour} follows. 

Eq.~\eqref{eq_probborder} can be obtained by a similar argument. In
this case we want
\begin{align*}
  &(\bra{\zeta_x^0}\bra{\zeta_{x\oplus
  1}^{2n-1}})\ket{\phi}=\frac{1}{d^{3/2}}\left(1+\phantom{\sum_{j=1}^{d-1}}\right.\\
&\left.\sum_{j=1}^{d-1}\exp\!\left[\pi i\!
  \left(\frac{2xj}{d}\!-\!\md_B\!\left[\frac{d\!-\!j}{d}\!\right]\!\frac{2n\!-\!1}{n}+\frac{2(d\!-\!j)(x\!\oplus\!1)}{d}\right)\!\right]\!\right)\\
&=\frac{1}{d^{3/2}}\left(\sum_{j=0}^{\lfloor d/2\rfloor}\exp\left[-\frac{\pi ij}{nd}\right]+\sum_{j=\lfloor d/2\rfloor+1}^{d-1}\!\!\exp\left[-\pi i
  \frac{j/d-1}{n}\right]\right)\\
&=\frac{-1}{d^{3/2}}\exp\left[-\frac{\pi
  i}{nd}(1+\lfloor\frac{d}{2}\rfloor)\right]\frac{1-\exp\left[-\frac{\pi
  i}{n}\right]}{1-\exp\left[-\frac{\pi
  i}{nd}\right]}\,
.
\end{align*}
Hence,
\begin{align*}
\sum_x|(\bra{\zeta_x^0}\bra{\zeta_{x\oplus 1}^{2n-1}})\ket{\phi}|^2&=\frac{1}{d^2}\frac{\sin^2\frac{\pi}{2n}}{\sin^2\frac{\pi}{2dn}}\, .
\end{align*}

Combining~\eqref{eq_probneighbour} and~\eqref{eq_probborder} we find
\begin{align*}
I_{n,d}(P_{XY|AB})=2n\left(1-\frac{\sin^2\frac{\pi}{2n}}{d^2\sin^2\frac{\pi}{2dn}}\right).
\end{align*} 
Using $x^2-x^4/3\leq\sin^2x\leq x^2$ for $0\leq x\leq 1$ leads to the
bound~\eqref{eq_Inbound}.

\section{Proof of Lemma~\ref{lem_extended}}\label{app_In} 

In the following we use the abbreviations $P_{XY|AB\lambda} \equiv
P_{XY|AB\Lambda}(\cdot,\cdot|\cdot,\cdot,\lambda)$ and
$P_{XY|ab\lambda} = P_{XY|AB\lambda}(\cdot,\cdot|a,b)$ for the
distributions conditioned on $\Lambda = \lambda$ and $(A,B) = (a,b)$.

The inequality in Lemma~\ref{lem_extended} can be expressed
in terms of the total variation distance, defined by $D(P_X,Q_X)
\equiv \frac{1}{2}\sum_x|P_X(x)-Q_X(x)|$, as
  \begin{align*}
   \int \mathrm{d}P_{\Lambda}(\lambda) D(P_{X |a_0 \lambda},  1/d)
   \leq\frac{d}{4} I_{n,d}(P_{XY|AB}) \, .
  \end{align*}
  where $1/d$ denotes the uniform distribution over $\{0, \ldots,
  d-1\}$, and where $a_0=0$.  Furthermore, using $P_{XY|AB} = \int
  \mathrm{d} P_{\Lambda}(\lambda) P_{XY|AB \lambda}$ (which holds
  because $P_{\Lambda|AB} = P_{\Lambda}$) and that $I_{n,d}$ is a
  linear function, we have
\begin{align*}
     I_{n,d}(P_{XY|AB}) = \int \mathrm{d}P_{\Lambda}(\lambda)
     I_{n,d}(P_{XY|AB\lambda}) \ .
\end{align*}
It therefore suffices to show that, for any $\lambda$, 
\begin{align*}
   D(P_{X |a_0 \lambda},  1/d) \leq\frac{d}{4} I_{n,d}(P_{XY|AB
     \lambda})  \ .
\end{align*}

For this, we consider the distribution $P_{X\oplus 1|a\lambda}$, which
corresponds to the distribution of $X$ if its values are shifted by
one (modulo $d$). According to Lemma~\ref{lem_disttoun} and using
$\frac{1}{d}\lfloor\frac{d^2}{4}\rfloor\leq\frac{d}{4}$ we have
\begin{align*}
  D(P_{X |a_0 \lambda},  1/d) \leq \frac{d}{4} D(P_{X\oplus 1|a_0
    \lambda},P_{X|a_0 \lambda}) \ .
\end{align*}
The assertion then follows with
    \begin{align*}
      &I_{n,d}(P_{XY|AB\lambda}) \\
      &=2n-\sum_x
      P_{XY|a_0b_0\lambda}(x,x\oplus 1)-\!\!\!\!\sum_{\genfrac{}{}{0pt}{}{x,a,b}{|a-b|=1}}P_{XY|ab\lambda}(x,x)    \\
      &\geq D(P_{X\oplus 1|a_0 b_0 \lambda},P_{Y|a_0 b_0 \lambda}) +\!\!\!\!
      \sum_{\genfrac{}{}{0pt}{}{a,b}{|a-b|=1}}D(P_{X|ab\lambda},P_{Y|ab\lambda})
      \\
      &\geq D(P_{X\oplus 1|a_0\lambda},P_{X|a_0\lambda}) \, ,
    \end{align*}
    where we have set $b_0 \equiv 2n-1$; the first inequality follows
    from Lemma~\ref{lem_Dbound}; the second is obtained with
    $P_{X|ab\lambda}=P_{X|a\lambda}$ and
    $P_{Y|ab\lambda}=P_{Y|b\lambda}$ (which are implied by the
    conditions stated in the lemma) as well as the triangle inequality
    for $D(\cdot, \cdot)$. \qed

\section{Additional Lemmas} \label{app_additional}

\begin{lemma} \label{lem_isometry} For any unit vectors $\wf{\psi},
  \wf{\psi'} \in \cH_1$ and $\wf{\psiu}, \wf{\psiu'} \in \cH_2$, where
  $\mathrm{dim}(\cH_1)\leq \mathrm{dim}(\cH_2)$ and
  $\braket{\psi}{\psi'} = \braket{\psiu}{\psiu'}$, there exists an
  isometry $U: \cH_1 \to \cH_2$ such that $U\wf{\psi} = \wf{\psiu}$
  and $U \wf{\psi'} = \wf{\psiu'}$.
\end{lemma}

\begin{proof}
  With $\alpha = \braket{\psi}{\psi'} = \braket{\psiu}{\psiu'}$ and
  $\beta = \sqrt{1-|\alpha|^2}$ we can write $\wf{\psi'} = \alpha
  \wf{\psi} + \beta \wf{\psi^\perp}$ and $\wf{\psiu'} = \alpha
  \wf{\psiu} + \beta \wf{\psiu^\perp}$ with unit vectors
  $\wf{\psi^\perp}$ and $\wf{\psiu^\perp}$ orthogonal to $\wf{\psi}$
  and $\wf{\psiu}$, respectively.  The isometry $U$ can be taken as
  any that acts as $\ketbra{\psiu}{\psi} +
  \ketbra{\psiu^\perp}{\psi^\perp}$ on the subspace spanned by
  $\wf{\psi}$ and $\wf{\psi'}$.
\end{proof}

\begin{lemma} \label{lem_Dbound} For two random variables $X$ and $Y$
  with joint distribution $P_{XY}$, the total variation distance
  between the marginal distributions $P_X$ and $P_Y$ satisfies
  \begin{equation*}
    D(P_X,P_Y)\leq 1-\sum_xP_{XY}(x,x)\, .
  \end{equation*}
\end{lemma}

\begin{proof}
  Consider $P_{XY}^{\neq} \equiv P_{XY|X\neq Y}$, the distribution of
  $X$ and $Y$ conditioned on the event that $X\neq Y$, as well as
  $P_{XY}^{=} \equiv P_{XY|X=Y}$ so that
  \begin{equation*}
    P_{XY}=p_{\neq}P_{XY}^{\neq}+(1-p_{\neq})P_{XY}^=
  \end{equation*}
  where $p_{\neq}\equiv 1-\sum_xP_{XY}(x,x)$. The marginals also obey this
  relation, i.e.,
\begin{eqnarray*}
  P_{X}&=&p_{\neq}P_{X}^{\neq}+(1-p_{\neq})P_{X}^=\\
  P_{Y}&=&p_{\neq}P_{Y}^{\neq}+(1-p_{\neq})P_{Y}^=\, .
\end{eqnarray*}
  Hence, since the total variation distance is convex,
  \begin{eqnarray*}
  D(P_X, P_Y)&\leq&p_{\neq} D(P_X^{\neq}, P_Y^{\neq}) + (1-p_{\neq})  D(P_X^=,
  P_Y^=) \\
  &\leq&p_{\neq} \, ,
  \end{eqnarray*}  
  where we have used the fact that the total variation distance is at
  most $1$, as well as $D(P_X^=,P_Y^=)=0$ in the last line.
\end{proof}

\begin{lemma}\label{lem_disttoun}
  The total variation distance between any probability distribution
  with range $\{0,1,\ldots,{d-1}\}$ and the uniform distribution over
  this set, $1/d$, is bounded by
  \begin{align*} 
  D(P_X,1/d)\leq\frac{1}{d}\lfloor\frac{d^2}{4}\rfloor  D(P_{X\oplus
    1},P_X) \ .
  \end{align*}
\end{lemma}

\begin{proof}
  Using $\frac{1}{d}\sum_{i=0}^{d-1}P_{X\oplus i}=1/d$ and the
  convexity of $D$, we find 
  \begin{align*}
    D(P_X,1/d)
    &=D\left(\frac{1}{d}\sum_{i=0}^{d-1}P_X,\frac{1}{d}\sum_{i=0}^{d-1}P_{X\oplus i}\right)\\
    &\leq\frac{1}{d}\sum_{i=0}^{d-1}D(P_X,P_{X\oplus i})\, .
  \end{align*}
  Because $D(P_{X\oplus(i-1)},P_{X\oplus i})=D(P_{X\oplus 1},P_X)$
  for all $i$ we have for $i\leq d/2$
  \begin{align*}
    D(P_X,P_{X\oplus
      i})&\leq D(P_X,P_{X\oplus(i-1)})+D(P_{X\oplus(i-1)},P_{X\oplus
      i})\\
&=D(P_X,P_{X\oplus(i-1)})+D(P_{X\oplus 1}, P_X) \ .
  \end{align*}
  Using this multiple times yields $D(P_X,P_{X\oplus i})\leq
  iD(P_{X\oplus 1}, P_X)$. Similarly, for $i\geq d/2$, we use
  \begin{align*}
    D(P_X,P_{X\oplus  i}) &\leq D(P_X,P_{X\oplus(i+1)})+D(P_{X\oplus(i+1)},P_{X\oplus
      i})\\
&=D(P_X,P_{X\oplus(i+1)})+D(P_{X\oplus 1}, P_X)
  \end{align*}
  multiple times to yield $D(P_X,P_{X\oplus i})\leq {(d-i)}D(P_{X\oplus
    1}, P_X)$.  Thus,
  \begin{multline*}
    \sum_{i=0}^{d-1}D(P_X,P_{X\oplus i})\\
\leq\left(\sum_{i=0}^{\lfloor d/2\rfloor}i+\sum_{i=\lfloor d/2\rfloor+1}^{d-1}(d-i)\right)D(P_{X\oplus
      1}, P_X)\\
=\left\lfloor\frac{d^2}{4}\right\rfloor D(P_{X\oplus 1}, P_X)\, .
  \end{multline*}
Combining this with the above concludes the proof. 
\end{proof}

Note that there are distributions that achieve the bound of
Lemma~\ref{lem_disttoun}, as can be seen for $d$ even and the
distribution $P_X=(2/d,2/d,\ldots,2/d,0,0,\ldots)$, for which
$D(P_X,1/d)=1/2$ and $D(P_{X\oplus 1}, P_X)=2/d$.

\end{document}